\documentclass[12pt,twoside,leqno]{article}
\usepackage{amsmath,amsthm,amssymb}
\numberwithin{equation}{section}
\usepackage[dvips]{graphicx,color,psfrag}

\makeatletter

\theoremstyle{plain}
\newtheorem{thm}{Theorem}[section]
\newtheorem{prop}[thm]{Proposition}
\newtheorem{lem}[thm]{Lemma}

\theoremstyle{definition}

\newtheorem{rem}[thm]{Remark}

\begin{document}
\title{Operator orderings and Meixner-Pollaczek polynomials}
\author{Genki Shibukawa}
\date{\empty}
\pagestyle{plain}

\maketitle

\begin{abstract}
The aim of this paper is to give identities which are generalizations of the formulas given by Koornwinder [J. Math. Phys. 30, (1989)] and Hamdi-Zeng [J. Math. Phys. 51, (2010)]. 
Our proofs are much simpler than and different from the previous investigations. 
 
\end{abstract}
\section{Introduction}
Let $W$ be the Weyl algebra generated by $p$ and $q$ with the relation $[p,q]:=pq-qp=1$. 
In this paper, we prove the following theorems.
\begin{thm}
\label{thm:theorem1}
We put $T:=pq+qp$. We obtain
\begin{align}
\label{eq:theorem1}
2^{n}\sum_{k=0}^{m}\binom{m}{k}p^{k}q^{n}p^{m-k}&=
2^{m}\sum_{k=0}^{n}\binom{n}{k}q^{k}p^{m}q^{n-k} \\
&=\begin{cases}
    2^{m}n!i^{-n}P_{n}^{(\frac{1+m-n}{2})}\left(\frac{i(T+m-n)}{2};\frac{\pi}{2}\right)p^{m-n} & (m\geq n) \\
    2^{n}m!i^{-m}q^{n-m}P_{m}^{(\frac{1+n-m}{2})}\left(\frac{i(T+n-m)}{2};\frac{\pi}{2}\right)  & (n\geq m)
  \end{cases}.\nonumber
\end{align}
In particular, we have(\cite{HZ}) 
\begin{equation}
\label{eq:theorem1prot}
\sum_{k=0}^{n}\binom{n}{k}p^{k}q^{n}p^{n-k}
=\sum_{k=0}^{n}\binom{n}{k}q^{k}p^{n}q^{n-k}
=n!i^{-n}P_{n}^{(\frac{1}{2})}\left(\frac{iT}{2};\frac{\pi}{2}\right).
\end{equation}
Here $P_{n}^{(\alpha)}(x;\phi)$ is the Meixner-Pollaczek polynomial given by the hypergeometric series
\begin{align}
\label{eq:Meixner-Pollaczek definition}
P_{n}^{(\alpha)}(x;\phi):=\frac{(2\alpha)_{n}}{n!}e^{in\phi}
{_{2}F_1}\left[\begin{matrix}-n,\alpha+ix\\
2\alpha \end{matrix};1-e^{-2i\phi}\right].
\end{align}
\end{thm}
\begin{thm}
\label{thm:theorem2}
Let $T_{m,n}$ be the sum of all possible terms containing $m$ factors of $p$ and $n$ factors of $q$. We have
\begin{align}
\label{eq:theorem2}
T_{m,n}= \begin{cases}
    \frac{n!}{2^{n}}\binom{m+n}{n}i^{-n}P_{n}^{(\frac{1+m-n}{2})}\left(\frac{i(T+m-n)}{2};\frac{\pi}{2}\right)p^{m-n} & (m\geq n) \\
    \frac{m!}{2^{m}}\binom{m+n}{m}i^{-m}q^{n-m}P_{m}^{(\frac{1+n-m}{2})}\left(\frac{i(T+n-m)}{2};\frac{\pi}{2}\right)  & (n\geq m)
  \end{cases}.
\end{align}
In particular, we have(\cite{K},\cite{HZ},\cite{FW})
\begin{equation}
\label{eq:original}
T_{n}:=T_{n,n}=\frac{n!}{2^{n}}\binom{2n}{n}i^{-n}P_{n}^{(\frac{1}{2})}\left(\frac{iT}{2};\frac{\pi}{2}\right).
\end{equation}
\end{thm}

The formula (\ref{eq:original}) for $T_{n}$ was first observed by Bender, Mead and Pinsky(\cite{BMP}), and proved by Koorwinder(\cite{K}).
The idea of the proof in \cite{K} is to consider the irreducible unitary representations of the Heisenberg group and some analysis for special functions. 
Moreover, a combinatorial proof was given by Hamdi and Zeng(\cite{HZ}). 
They used the rook placement interpretation of the normal ordering of the Weyl algebra 
and gave also a proof of (\ref{eq:theorem1prot}), which was first observed by \cite{BD}. 
Our results extend these to general $m$ and $n$.

The proofs given in this paper are much simpler than the investigations(\cite{K}, \cite{BD}). 
Actually, we only use some basic properties of the Weyl algebra and a certain transformation formula of the hypergeometric function. 
Our proofs clarify the reason why (\ref{eq:theorem1prot}) and (\ref{eq:original}) are equal up to constant, which is not explained in \cite{HZ}.
\section{Proof of Theorem \ref{thm:theorem1}}
The operations $L_{A},R_{A} \in \mathrm{End}\,_{\mathbf{C}}(W)$ are respectively left and right multiplications, that is, 
\begin{equation}
L_{A}.X:=AX,\,\,\,\,\,R_{A}.X:=XA, \,\,\,\,\,(A,X \in W).
\end{equation}
We introduce some useful operators(\cite{W}). 
\begin{equation}
\check{\rm{ad}}(A):=L_{A}+R_{A}.
\end{equation}
We remark that $L, R:W \rightarrow \mathrm{End}\,_{\mathbf{C}}(W)$ are linear, hence $\check{\rm{ad}}$ is also linear. 
In addition, since $\check{\rm{ad}}(A)^{N}.1=2^{N}A^{N}$, we obtain the following lemma immediately.
\begin{lem}
\label{lem:primitive key lemma}
Let $t_{1},\cdots,t_{n}$ be indeterminates. For any $N \in \mathbf{Z}_{\geq 0}$, we obtain
\begin{equation}
\label{eq:primitive key lemma}
\left\{\sum_{k=1}^{n}t_{k}\check{\rm{ad}}(A_{k})\right\}^{N}\!\!\!\!.1=2^{N}\left\{\sum_{k=1}^{n}t_{k}A_{k}\right\}^{N}.
\end{equation}
In particular, we have 
\begin{equation}
\label{eq:primitive key lemma2}
(t_{1}\check{\rm{ad}}(p)+t_{2}\check{\rm{ad}}(q))^{N}.1=2^{N}(t_{1}p+t_{2}q)^{N}.
\end{equation}
\end{lem}
\begin{rem}
When $N=n$ in Lemma\,\ref{lem:primitive key lemma}, comparing the coefficients of $t_{1}\cdots t_{n}$ on both sides of the (\ref{eq:primitive key lemma}), we obtain the following formula immediately. 
\begin{equation}
F(\check{\rm{ad}}(\underline{A}_{n})).1=2^{n}F(\underline{A}_{n}).
\end{equation}
Here, $\underline{A}_{n}:=\,(A_{1},\cdots,A_{n}),\,\,\check{\rm{ad}}(\underline{A}_{n}):=\,(\check{\rm{ad}}(A_{1}),\cdots,\check{\rm{ad}}(A_{n}))$ and 
\begin{align}
F(\underline{A}_{n}):=\sum_{\sigma \in \mathfrak{S}_{n}}A_{\sigma(1)}\cdots A_{\sigma(n)},\,\,\,\,
F(\check{\rm{ad}}(\underline{A}_{n})):=\sum_{\sigma \in \mathfrak{S}_{n}}\check{\rm{ad}}(A_{\sigma(1)})\cdots \check{\rm{ad}}(A_{\sigma(n)}).
\end{align}
\end{rem}
\begin{lem}
\label{prop:commutant}
The operators $\check{\rm{ad}}(p)$ and $\check{\rm{ad}}(q)$ are commutative.
\end{lem}
\begin{proof}
Obviously $L_{A}$ and $R_{B}$ are commutative. 
Since $L$ is a homomorphism and $R$ is an anti-homomorphism, we have 
\begin{align}
[\check{\rm{ad}}(p),\check{\rm{ad}}(q)]
=[L_{p}+R_{p},L_{q}+R_{q}]
=[L_{p},L_{q}]+[R_{p},R_{q}]
=L_{pq-qp}-R_{pq-qp}=0. \nonumber
\end{align}
\end{proof}
\begin{prop}
\label{prop:basic expan}
\begin{equation}
\label{eq:basic expan}
\check{\rm{ad}}(p)^{m}\check{\rm{ad}}(q)^{n}.1=
2^{n}\sum_{k=0}^{m}\binom{m}{k}p^{k}q^{n}p^{m-k}
=2^{m}\sum_{k=0}^{n}\binom{n}{k}q^{k}p^{m}q^{n-k}. 
\end{equation}
\begin{proof}
Since $L_{A}$ and $R_{B}$ are commutative, $L$ is a homomorphism and $R$ is an anti-homomorphism, we obtain 
\begin{align}
\check{\rm{ad}}(p)^{m}\check{\rm{ad}}(q)^{n}.1
=(L_{p}+R_{p})^{m}.2^{n}q^{n}
=2^{n}\sum_{k=0}^{m}\binom{m}{k}L_{p^{k}}R_{p^{m-k}}.q^{n}
=2^{n}\sum_{k=0}^{m}\binom{m}{k}p^{k}q^{n}q^{m-k}. \nonumber
\end{align}
On the other hand, since $\check{\rm{ad}}(p)$ and $\check{\rm{ad}}(q)$ are commutative, we have
\begin{equation}
\check{\rm{ad}}(p)^{m}\check{\rm{ad}}(q)^{n}=\check{\rm{ad}}(q)^{n}\check{\rm{ad}}(p)^{m}. \nonumber 
\end{equation}
Hence, the second equality of (\ref{eq:basic expan}) can be proved in the same way.
\end{proof}
\begin{rem}
Wakayama(\cite{W}) has constructed the oscillator representation of the simple Lie algebra $\mathfrak{sl}_{2}$ by $\check{\rm{ad}}$ and $\rm{ad}$ in $\mathrm{End}\,_{\mathbf{C}}(W)$ and 
then, proves that $\check{\rm{ad}}(p)^{n}\check{\rm{ad}}(q)^{n}.1$ satisfies the difference equation of the Meixner-Pollaczek polynomials. 
\end{rem}
\end{prop}
Since $T=pq+qp$ and $pq-qp=1$, we have
\begin{equation}
pq=\frac{T+1}{2},\,\,\,\,\,qp=\frac{T-1}{2}.
\end{equation}
The proof of the following lemma is straightforward.
\begin{lem}
\label{lem:basic properties}
{\rm{(1)}}\,
Let $f(T) \in \mathbf{C}[T], l \in \mathbf{Z}_{\geq 0}$. We have
\begin{equation}
\label{eq:commutation relation1}
p^{l}f(T)=f(T+2l)p^{l},\,\,\,\,\,q^{l}f(T)=f(T-2l)q^{l}.
\end{equation}
{\rm{(2)}}\,
For any $l \in \mathbf{Z}_{\geq 0}$, we have
\begin{equation}
\label{eq:factor1}
p^{l}q^{l}=\left(\frac{1+T}{2}\right)_{l},\,\,\,\,\,q^{l}p^{l}=(-1)^{l}\left(\frac{1-T}{2}\right)_{l}.
\end{equation}
Here, $(x)_{l}:=x(x+1)\cdots(x+l-1),\,(x)_{0}:=1$.
\end{lem}
\begin{prop}
\begin{equation}
\label{eq:generating function}
n!i^{-n}P_{n}^{(\alpha)}\left(\frac{ix}{2};\frac{\pi}{2}\right)
=\sum_{k=0}^{n}\binom{n}{k}(-1)^{k}\left(\alpha-\frac{x}{2}\right)_{k}\left(\alpha+\frac{x}{2}\right)_{n-k}.
\end{equation}
\end{prop}
\begin{proof}
It follows from the formula (2.3.14) in \cite{AAR} that
\begin{equation}
\label{eq:generating function prot}
(LHS)=(2\alpha)_{n}{_{2}F_1}\left[\begin{matrix}-n,\alpha-\frac{x}{2}\\2\alpha \end{matrix};2\right]
=\left(\alpha+\frac{x}{2}\right)_{n}{_{2}F_1}\left[\begin{matrix}-n,\alpha-\frac{x}{2}\\-n-\alpha-\frac{x}{2}+1 \end{matrix};-1\right]
=(RHS).\nonumber
\end{equation}
\begin{rem}
One may also prove this proposition using the generating function for Meixner-Pollaczek polynomials.
\end{rem}
\end{proof}
We now prove Theorem \ref{thm:theorem1} as follows. 
If $m\geq n$, 
\begin{align}
2^{m}\sum_{k=0}^{n}\binom{n}{k}q^{k}p^{m}q^{n-k}
&=2^{m}\sum_{k=0}^{n}\binom{n}{k}q^{k}p^{k}p^{m-n}p^{n-k}q^{n-k} \nonumber \\
&=2^{m}\sum_{k=0}^{n}\binom{n}{k}(-1)^{k}\left(\frac{1-T}{2}\right)_{k}p^{m-n}\left(\frac{1+T}{2}\right)_{n-k} \nonumber \\
&=2^{m}\sum_{k=0}^{n}\binom{n}{k}(-1)^{k}\left(\frac{1-T}{2}\right)_{k}\left(\frac{1+T}{2}+m-n\right)_{n-k}p^{m-n} \nonumber \\
&=2^{m}n!i^{-n}P_{n}^{\left(\frac{1+m-n}{2}\right)}\left(\frac{i(T+m-n)}{2};\frac{\pi}{2}\right)p^{m-n}. \nonumber
\end{align}
The second equality follows from (\ref{eq:factor1}), the third from (\ref{eq:commutation relation1}) and the fourth from (\ref{eq:generating function}).
By Proposition\,\ref{prop:basic expan}, the case of $n\geq m$ can be proved in the same way.

\section{Proof of Theorem \ref{thm:theorem2}}
Comparing the coefficients of $t_{1}^{m}t_{2}^{n}$ on both sides in (\ref{eq:primitive key lemma2}) for $N=m+n$, one obtain the key Proposition.
\begin{prop}
\label{prop:expression of T_{m,n}}
For any $m,n \in \mathbf{N}$, we have
\begin{equation}
\label{eq:expression of T_{m,n}}
T_{m,n}=\frac{1}{2^{m+n}}\frac{(m+n)!}{m!n!}\check{\rm{ad}}(p)^{m}\check{\rm{ad}}(q)^{n}.1.
\end{equation}
\end{prop}
Theorem\,\ref{thm:theorem2} follows immediately from (\ref{eq:expression of T_{m,n}}), (\ref{eq:basic expan}) and (\ref{eq:theorem1}).
\begin{rem}
{(\rm{1})}\,
If $m\geq n$, then we have the following result immediately by Theorem\,\ref{thm:theorem2} and (\ref{eq:factor1}).
\begin{equation}
T_{m,n}q^{m-n}=\frac{n!}{2^{n}}\binom{m+n}{n}i^{-n}\left(\frac{1+T}{2}\right)_{m-n}P_{n}^{(\frac{1+m-n}{2})}\left(\frac{i(T+m-n)}{2};\frac{\pi}{2}\right).
\end{equation}
The case of $n\geq m$ is similar.\\
{(\rm{2})}\,
If $m\geq n$, then a explicit expression of the Poincare-Birkhoff-Witt theorem for $T_{m,n}$ follows from (\ref{eq:theorem2}), (\ref{eq:Meixner-Pollaczek definition}) and (\ref{eq:factor1}). 
\begin{align}
T_{m,n}=\frac{1}{2^{n}}\frac{m!}{(m-n)!}\binom{m+n}{n}\sum_{k\geq 0}\binom{n}{k}\frac{2^{k}}{(1+m-n)_{k}}q^{k}p^{k+m-n}. 
\end{align}
The case of $n\geq m$ is similar.
\end{rem}

Recently, a generalization of Theorem\,\ref{thm:theorem2} using the multivariate Meixner-Pollaczek polynomials in the framework of the Gelfand pair has been established in \cite{FW}. 
Another proof of \cite{FW} in our current approach would be desirable. 
\section*{Acknowledgment}
The author would like to thank Professors Masato Wakayama and Hiroyuki Ochiai for many helpful comments. 
This work has been supported by the JSPS Research Fellowship.
\bibliographystyle{amsplain}

\noindent Graduate School of Mathematics, Kyushu University\\
744, Motooka, Nishi-ku, Fukuoka, 819-0395, JAPAN.\\
E-mail: g-shibukawa@math.kyushu-u.ac.jp

\end{document}